\documentclass[12pt]{article}
\usepackage{amssymb,amsthm,amsmath,latexsym}
\usepackage{url}
\usepackage{color}
\newtheorem{thm}{Theorem}
\newtheorem{prop}[thm]{Proposition}
\newtheorem{lem}[thm]{Lemma}
\newtheorem{cor}[thm]{Corollary}

\theoremstyle{remark}
\newtheorem{rem}[thm]{Remark}
\newcommand{\FF}{\mathbb{F}}

\newcommand{\0}{\mathbf{0}}

\DeclareMathOperator{\wt}{wt}

\begin{document}
\title{Remark on subcodes of linear complementary dual codes}

\author{
Masaaki Harada\thanks{
Research Center for Pure and Applied Mathematics,
Graduate School of Information Sciences,
Tohoku University, Sendai 980--8579, Japan.
email: {\tt mharada@tohoku.ac.jp}.}
and 
Ken Saito\thanks{
Research Center for Pure and Applied Mathematics,
Graduate School of Information Sciences,
Tohoku University, Sendai 980--8579, Japan.
email: {\tt kensaito@ims.is.tohoku.ac.jp}.}
}


\maketitle

\begin{abstract}
We show that
any ternary Euclidean (resp.\ quaternary Hermitian)
linear complementary dual
$[n,k]$ code contains a 
Euclidean (resp.\ Hermitian)
linear complementary dual $[n,k-1]$ subcode
for $2 \le k \le n$.
As a consequence, we derive a bound on the largest minimum weights
among ternary Euclidean
linear complementary dual codes and quaternary Hermitian 
linear complementary dual codes.
\end{abstract}

\section{Introduction}\label{Sec:1}

Let $\FF_q$ denote the finite field of order $q$,
where $q$ is a prime power.
An $[n,k]$ code $C$ over $\FF_q$ is a $k$-dimensional vector 
subspace of $\FF_q^n$.
A code over $\FF_q$ $(q=2,3 \text{ and }4)$ is called {\em binary, ternary}
and {\em quaternary}, respectively.
The elements of $C$ are called {\em codewords} and the {\em weight}
$\wt(x)$ of a codeword $x$ is the number of non-zero coordinates
in $x$. The {\em minimum weight} of $C$ is defined as 
$\min\{\wt(x) \mid \0_n \ne x \in C\}$,
where $\0_n$ denotes the zero vector of length $n$.

The {\em Euclidean dual} code $C^{\perp}$ of an $[n,k]$ code $C$ 
over $\FF_q$ is defined as
$
C^{\perp}=
\{x \in \FF_q^n \mid \langle x,y\rangle_E = 0 \text{ for all } y \in C\},
$
where $\langle x,y\rangle_E = \sum_{i=1}^{n} x_i {y_i}$
for $x=(x_1,\ldots,x_n), y=(y_1,\ldots,y_n) \in \FF_q^n$.
For any element $x \in \FF_{q^2}$, the conjugation of $x$ is
defined as $\overline{x}=x^q$.
The {\em Hermitian dual} code $C^{\perp_H}$ of an $[n,k]$ code $C$ 
over $\FF_{q^2}$ is defined as
$
C^{\perp_H}=
\{x \in \FF_{q^2}^n \mid \langle x,y\rangle_H = 0 \text{ for all } y \in C\},
$
where $\langle x,y\rangle_H= \sum_{i=1}^{n} x_i \overline{y_i}$
for $x=(x_1,\ldots,x_n), y=(y_1,\ldots,y_n) \in \FF_{q^2}^n$.
A code $C$ over $\FF_q$ is called 
{\em Euclidean linear complementary dual}
if $C \cap C^\perp = \{\0_n\}$.
A code $C$ over $\FF_{q^2}$ is called 
{\em Hermitian linear complementary dual}
if $C \cap C^{\perp_H} = \{\0_n\}$.
These two families of codes are collectively called
{\em linear complementary dual} (LCD for short) codes.

LCD codes were introduced by Massey~\cite{Massey} and gave an optimum linear
coding solution for the two user binary adder channel.
Recently, much work has been done concerning LCD codes
for both theoretical and practical reasons.
In particular, Carlet, Mesnager, Tang, Qi and Pellikaan~\cite{CMTQ2}
showed that 
any code over $\FF_q$ is equivalent to some Euclidean LCD code
for $q \ge 4$ and
any code over $\FF_{q^2}$ is equivalent to some Hermitian LCD code
for $q \ge 3$.
This motivates us to study Euclidean LCD codes over $\FF_q$
$(q=2,3)$ and quaternary Hermitian LCD codes.

It is a fundamental problem to determine the largest minimum
weight among all codes in a certain class of codes for a given pair $(n,k)$.
Let $d^E_q(n,k)$ denote the largest minimum weight among
all Euclidean LCD $[n,k]$ codes over $\FF_q$.
Let $d^H_{q^2}(n,k)$ denote the largest minimum weight among
all Hermitian LCD $[n,k]$ codes over $\FF_{q^2}$.
For codes $C$ and $D$ over $\FF_q$,
we say that $C$ is {\em subcode} of $D$ if $C \subset D$.
Galvez, Kim, Lee, Roe and Won~\cite{bound} showed that
any binary Euclidean LCD $[n,k]$ code contains a Euclidean
LCD $[n,k-1]$ subcode
for odd integers $k$ with $k \ge 3$.
As a consequence, it can be easily shown that
$d^E_2(n,k) \le d^E_2(n,k-1)$ for odd $k \ge 3$.
Then Carlet, Mesnager, Tang and Qi~\cite{CMTQ} showed that
$d^E_2(n,k) \le d^E_2(n,k-1)$ for even integers $k$ with $k \ge 2$.

The main aim of this note is to establish the following
theorem.

\begin{thm}\label{thm:main}
Suppose that $2 \le k \le n$.  Then
\[
d^E_3(n,k) \le d^E_3(n,k-1) \text { and }
d^H_4(n,k) \le d^H_4(n,k-1).
\]
\end{thm}

\section{Definitions, notations and basic results}


For any element $x \in \FF_{q^2}$, the conjugation of $x$ is
defined as $\overline{x}=x^q$.
Let $A^T$ denote the transpose of a matrix $A$.
For a matrix $A=(a_{ij})$, 
the conjugate matrix of $A$ is defined as
$\overline{A}=(\overline{a_{ij}})$.
A matrix whose rows are linearly
independent and generate a code $C$ over $\FF_q$ is 
called a {\em generator matrix} of $C$. 
The following characterization is due to Massey~\cite{Massey}.

\begin{lem}\label{lem:LCD}
Let $C$ be a code over $\FF_q$ (resp.\ $\FF_{q^2}$).  
Let $G$  be a generator matrix of $C$.
Then the following properties are equivalent:
\begin{itemize}
\item[\rm (i)] $C$ is  Euclidean (resp.\ Hermitian) LCD,
\item[\rm (ii)] $C^\perp$ is  Euclidean LCD  
(resp.\ $C^{\perp_H}$ is Hermitian LCD),
\item[\rm (iii)] $G G^T$ (resp.\ $G \overline{G}^T$) is nonsingular.
\end{itemize}
\end{lem}

A code $C$ over $\FF_q$ is called 
{\em Euclidean self-orthogonal}
if $C \subset C^\perp$.
A code $C$ over $\FF_{q^2}$ is called 
{\em Hermitian self-orthogonal}
if $C \subset C^{\perp_H}$.
It is trivial that 
a code $C$ over $\FF_q$ is Euclidean self-orthogonal if and only if
$G G^T=O$ for a generator matrix $G$ of $C$, where $O$ is the
zero matrix.
It is also trivial that 
a code $C$ over $\FF_{q^2}$ is Hermitian self-orthogonal if and only if
$G \overline{G}^T=O$ for a generator matrix $G$ of $C$.

\begin{lem}[see {\cite[Theorem~1.4.10]{HP}}]
\label{lem:SO}
\begin{itemize}
\item[\rm (i)]
Let $C$ be a ternary code.
Then $C$ is Euclidean self-orthogonal if and only if
the weights of all codewords of $C$ are divisible by three.
\item[\rm (ii)]
Let $C$ be a quaternary code.
Then $C$ is Hermitian self-orthogonal if and only if
the weights of all codewords of $C$ are even.
\end{itemize}
\end{lem}

\section{Proof of Theorem~\ref{thm:main}}

As a consequence of the following proposition, we immediately have 
Theorem~\ref{thm:main}.

\begin{prop}\label{prop:main}
Suppose that $2 \le k \le n$.
\begin{itemize}
\item[\rm (i)]
Any ternary  Euclidean LCD $[n,k]$ code contains a
Euclidean LCD $[n,k-1]$ subcode.
\item[\rm (ii)]
Any quaternary Hermitian LCD $[n,k]$ code contains a 
Hermitian LCD $[n,k-1]$ subcode.
\end{itemize} 
\end{prop}
\begin{proof}
The proofs of assertions (i) and (ii) are similar.
In order that we simultaneously give the proofs, 
we employ the terms and notations listed in Table~\ref{Tab}.

\begin{table}[thb]
\caption{Terms and notations}
\label{Tab}
\begin{center}
{\small
\begin{tabular}{c|cc}
\noalign{\hrule height0.8pt}
 & (i) & (ii) \\ 
\hline
$q$ & $3$ & $4$ \\
$p$ & $3$ & $2$ \\
LCD & Euclidean LCD & Hermitian LCD \\
self-orthogonal & Euclidean self-orthogonal & Hermitian self-orthogonal \\
$\langle x,x\rangle$ & $\langle x,x\rangle_E$ & $\langle x,x\rangle_H$ \\
$G^*$ & $G^T$ & $\overline{G}^T$\\
\noalign{\hrule height0.8pt}
\end{tabular}
}
\end{center}
\end{table}

Let $C$ be an LCD $[n,k]$ code over $\FF_q$.
Since $C$ is  LCD, $C$ is not self-orthogonal.
By Lemma~\ref{lem:SO}, 
there is a codeword $x$ of $C$ with $\wt(x) \not\equiv 0 \pmod p$.
Note that $\langle x,x\rangle \ne 0$.
We may assume without loss of generality that the generator matrix of
$C$ has one of  the following forms:
\[
G_1=
\left(\begin{array}{ccc}
&x& \\
&y_2& \\
& \vdots & \\
&y_\ell& \\
&z_{\ell+1}& \\
& \vdots & \\
&z_k& \\
\end{array}\right),
G_2=
\left(\begin{array}{ccc}
&x& \\
&y_2& \\
& \vdots & \\
&y_k& \\
\end{array}\right),
G_3=
\left(\begin{array}{ccc}
&x& \\
&z_{2}& \\
& \vdots & \\
&z_k& \\
\end{array}\right),
\]
where 
$\langle x,y_i\rangle =0$, 
$\langle x,z_i\rangle =1$ and $2 \le \ell \le k-1$.
For the matrices $G_1$ and $G_3$, we consider the following matrices:
\[
G'_1=
\left(\begin{array}{ccc}
&x& \\
&y_2& \\
& \vdots & \\
&y_\ell& \\
&z_{\ell+1}+(p-1)z_k& \\
& \vdots & \\
&z_{k-1}+(p-1)  z_k& \\
&z_k+(p-1)  \langle x,x\rangle x & 
\end{array}\right)\text{ and }
G'_3=
\left(\begin{array}{ccc}
&x& \\
&z_{2}+(p-1)  z_k& \\
& \vdots & \\
&z_{k-1}+(p-1)  z_k& \\
&z_k+(p-1) \langle x,x\rangle x & 
\end{array}\right),
\]
respectively.
Note that 
$
\langle x, z_k+(p-1)  \langle x,x\rangle x \rangle =
0
$.
Therefore,  $C$ has generator matrix $G$ of  the following form:
\[
G=
\left(\begin{array}{ccc}
&x& \\
&G_0& \\
\end{array}\right)\text{ and }
G_0=
\left(\begin{array}{ccc}
&y_2& \\
& \vdots & \\
&y_k& \\
\end{array}\right),
\]
satisfying  that
$\langle x,y_i\rangle =0$.
Then we have
\begin{align*}
GG^*
&=
\left(\begin{array}{cccc}
\langle x, x\rangle  & 0 & \cdots & 0 \\
0 & & & \\
\vdots & &
G_0G_0^* \\
0 & & &
\end{array}\right).
\end{align*}
Since $\langle x, x\rangle \ne 0$ and 
$\det GG^* \ne 0$,
it follows that
$\det G_0G_0^* \ne 0$.
By Lemma~\ref{lem:LCD}, the matrix $G_0$ is a generator matrix of an
LCD $[n,k-1]$ code $C_0$ over $\FF_q$ with $C_0 \subset C$.
\end{proof}

\section{Remarks}

As another consequence of Proposition~\ref{prop:main}, we have the following:

\begin{prop}\label{prop:r1}
Suppose that $1 \le k \le n-1$.
\begin{itemize}
\item[\rm (i)]
For any ternary Euclidean LCD $[n,k]$ code $C$,
there is a  Euclidean  LCD $[n,k+1]$ code containing $C$ as a subcode.
\item[\rm (ii)]
For any quaternary Hermitian LCD $[n,k]$ code $C$,
there is a Hermitian LCD $[n,k+1]$ code containing $C$ as a subcode.
\end{itemize}
\end{prop}
\begin{proof}
The proofs of assertions (i) and (ii) are similar.
In order that we simultaneously give the proofs, 
we employ the terms and notations listed in Table~\ref{Tab}.
In addition, we denote $C^\perp$ and $C^{\perp_H}$ by
$C^*$ for (i) and (ii), respectively.
Let $C$ be an LCD $[n,k]$ code over $\FF_q$.
By Lemma~\ref{lem:LCD}, $C^*$ is LCD.
By Proposition~\ref{prop:main},
there is an LCD $[n,n-k-1]$ code $E$ with $E \subset C^*$.
Again by Lemma~\ref{lem:LCD}, $E^*$ is an LCD $[n,k+1]$ code.
Since $E \subset C^*$, 
$E^*$ is an LCD code with $C \subset E^*$.
The result follows.
\end{proof}

\begin{rem}
We give an alternative proof of the above proposition,
which is constructive.
Let the notations be as above.
In addition, let $G$ be a generator matrix of $C$.
By Lemma~\ref{lem:LCD},
$C^*$ is LCD.
Thus, $C^*$ is not self-orthogonal.
Hence, there is a nonzero vector $x$ of $C^*$ with 
$\wt(x) \not\equiv 0 \pmod p$.
Note that $x \not\in C$ and $\langle x, x\rangle \ne 0$.
Consider the following matrix:
\[
G'=
\left(\begin{array}{ccc}
&x& \\
&G& \\
\end{array}\right).
\]
Then we have
\begin{align*}
G' G'^*
=
\left(\begin{array}{cccc}
\langle x, x\rangle  & 0 & \cdots & 0 \\
0 & & & \\
\vdots & &
GG^* \\
0 & & &
\end{array}\right).
\end{align*}
Since $\det GG^* \ne 0$,
it follows that 
$\det G' G'^* \ne 0$.
By Lemma~\ref{lem:LCD}, the matrix $G'$ is a generator matrix of an
LCD $[n,k+1]$ code $C'$ over $\FF_q$ with $C \subset C'$.
\end{rem}

As described in Section~\ref{Sec:1},
any code over $\FF_q$ is equivalent to some Euclidean LCD code
for $q \ge 4$ and
any code over $\FF_{q^2}$ is equivalent to some Hermitian LCD code
for $q \ge 3$~\cite{CMTQ2}.
Hence, for any $[n,k-1]$ subcode $D$ of a 
Euclidean (resp.\ Hermitian) LCD $[n,k]$ code
over $\FF_q$ ($q >3$) (resp.\ $\FF_{q^2}$ ($q >2$)),
there is a Euclidean (resp.\ Hermitian) LCD $[n,k-1]$ code 
$E$ such that $D$ is equivalent to $E$.
By Theorem~\ref{thm:main}, we have the following:

\begin{cor}
Suppose that $2 \le k \le n$.  Then
\[
d^E_q(n,k) \le d^E_q(n,k-1) \text{ and }
d^H_{q^2}(n,k) \le d^H_{q^2}(n,k-1).
\]
\end{cor}



\end{document}